\RequirePackage{fix-cm}
\documentclass[smallextended]{svjour3}       
\smartqed  
\usepackage{amsfonts}
\usepackage{amssymb}
\usepackage{latexsym}
\usepackage{enumerate}
\usepackage{amsmath}
\usepackage{booktabs}
\usepackage{multirow}
\usepackage{bbm}
\usepackage{subfloat}
\usepackage{subfig}
\usepackage{float}
\usepackage[export]{adjustbox}
\usepackage{color}
\usepackage{graphicx}
\usepackage{mathptmx}      
\usepackage{ulem}
\usepackage{bm}

\usepackage{anysize} 
 
\begin{document}

\title{Asset allocation: new evidence through network approaches}



\author{Gian Paolo Clemente      \and
        Rosanna Grassi      \and 
        Asmerilda Hitaj
}


\institute{G. P. Clemente\thanks{Corresponding author.}  \at
             Catholic University of Milan, Department of Mathematics, Finance and Econometrics. 
              \email{gianpaolo.clemente@unicatt.it}           
           \and
            R. Grassi   \at
              University of Milano - Bicocca, Department of Statistics and Quantitative Methods. 
              \email{rosanna.grassi@unimib.it}
             \and
             A. Hitaj   \at
           University of Milano - Bicocca, Department of Statistics and Quantitative Methods. 
           \email{asmerilda.hitaj1@unimib.it}           
}

\date{Received: date / Accepted: date}

\maketitle

\begin{abstract}
The main contribution of the paper is to employ the financial market network as a useful tool to improve the portfolio selection process, where nodes indicate securities and edges capture the dependence structure of the system. 
Three different methods are proposed in order to extract the dependence structure between assets in a network context. Starting from this modified structure, we formulate and then we solve the asset allocation problem. 
We find that the portfolios obtained through a network-based approach are composed mainly of peripheral assets, which are poorly connected with the others. These portfolios, in the majority of cases, are characterized by an higher trade-off between performance and risk with respect to the traditional Global Minimum Variance (GMV) portfolio. 
Additionally, this methodology benefits of a graphical visualization of the selected portfolio directly over the graphic layout of the network, which helps in improving our understanding of the optimal strategy. 
\end{abstract}

\vspace{4mm}
\noindent
\textbf{Keywords:}
Portfolio Selection, Networks, Global minimum variance, Dependence structure 
\newline
\textbf{JEL Classification: G11, C6}

\section{Introduction}
Modern portfolio theory, which originated with Harry Markowitz's seminal paper in 1952 (see \cite{markowitz52}) has stood the test of time and continues to be the intellectual foundation for real-world portfolio management. In this static framework, investors optimally allocate their wealth across a set of assets considering only the first and second moment of the returns' distribution. Despite the profound changes derived from this publication, the out-of-sample performance of Markowitz's prescriptions is often not as promising as expected. The poor performance of Markowitz's rule stems from the large estimation errors on the vector of expected returns (see \cite{merton80}) and on the covariance matrix (see \cite{jobson80}) leading to the well-documented error-maximizing property discussed in \cite{michaud2007estimation}. The magnitude of this problem is evident when we acknowledge the modest improvements achieved by those models specifically designed to tackle the estimation risk (see \cite{DeMiguel07}). Moreover, the evidence indicates that the simple yet effective equally-weighted portfolio rule has not been consistently out-performed by more sophisticated alternatives, as reported in \cite{bloomfield1977portfolio,DeMiguel07}.\\  
\noindent
The literature on portfolio selection has been extended on several directions. On one hand, some extensions modify the optimal problem, 
considering higher moments of returns' distribution (see \cite{martellini10} and the reference therein), exploiting alternative risk measures (see \cite{Campbell01,Uryasev02}) and as well as  utility functions (see \cite{tversky1992advances,HeCPT}), or proposing dynamic approaches (see among others \cite{Yin04,Brandt06}).  On the other hand, there is now a vast literature on how to deal with the problem, that focuses on improved estimation procedures, considering that the optimal allocation is very sensitive to the estimation of moments and co-moments\footnote{This sensitivity has generally been attributed to the tendency of the optimization to magnify the effects of estimation error. Michaud in \cite{michaud2007estimation} referred to \lq\lq portfolio optimization \rq\rq as \lq\lq error maximization \rq\rq. 
	Efforts to improve parameters estimation procedure include among others the papers \cite{ledoit03,martellini10}. Empirical analyses have shown that the use of improved estimators for moments and co-moments leads to higher out-of-sample performance compared to the sample estimation, see among others \cite{hitajBeta}.}. \\
Yet, portfolio selection problem has attracted the attention of operation research experts, in order to improve the Markowitz model (see \cite{Cesarone2013}, \cite{Boginski2014}).

We manage the issue under a different perspective, by proposing an optimization problem based on a network structure.
Precisely, we tackle the asset allocation problem considering our portfolio as a network in which firms/stocks are identified by nodes and the links represent the dependence structure of returns. In this way, we aim at exploiting the underlying structure of this financial market network as an effective tool in enhancing the portfolio selection process.\\ 
The use of network theory in financial framework is not completely new in the literature, although its use has been extensively developed over the last decade. Most of existing research using network theory concentrates on issues such as financial stability and contagion, specifically referring to systemic risk (for a recent survey we can refer to \cite{Neveu2018} or \cite{Caccioli2018}.) \\
As regard to  the application of network theory in portfolio selection process, Mantegna \cite{Mantegna99} proposed to build a network where edges' weights are inversely related to correlation. By computing a minimum spanning tree (MST), the author provides a way to reduce the network complexity in terms of number of edges. Onnela et al. \cite{Onnela2003} studied the time dependence of the ‘‘asset tree’’ in order to reflect the financial market taxonomy. In particular, the authors point out that the diversification aspect of portfolio optimization results in the fact that the assets of the classic Markowitz portfolio are always located on the outer leaves of the tree. Pozzi et al. \cite{pozzi13} investigate how financial filtered networks, based on MST and Planar Maximally Filtered Graphs (PMFG), can be used to characterize the heterogeneous spreading of risk across a financial market and how this information can be employed to reduce investment risk by constructing well-diversified portfolios. All these methods aim at identifying a portfolio of stocks, with a given cardinality, directly from the dependency structure provided by the financial filtered networks. 
Although these papers cover this innovative way, investigating  how  network topology could be used as an effective tool in  the portfolio selection process, they allow to reduce the complexity of the network but they do not provide a network-based criterion for assessing portfolios' optimal weights.\\
Peralta and Zaarei \cite{peralta16} establish a bridge between Markowitz's framework and network theory, showing a negative relationship between optimal portfolio weights and the centrality of assets in the financial market network. 
Boginski et al. \cite{Boginski2014} propose a method based on network clusters for selecting diversified portfolios. They formulate the portfolio selection problem as a particular clique relaxation problem in the market graph.
 \color{black} 
Our contribute to the literature is instead to propose a new network-based approach to deal with the problem of optimal asset allocation, which is alternative to classical portfolio strategies. 
In line with this kind of literature, the starting point is the use of correlated network, where edges' weights are directly related to linear correlation coefficients between each couple of assets. In order to develop an efficient investment strategy that can benefit from the knowledge of such market dependency structure, we optimize an objective function that takes into account, through an interconnectedness matrix, the community structure of each node in the system. We catch how much a node is embedded in the system, by adapting to this context the clustering coefficient, a specific network index, meaningful in financial literature to assess systemic risk (see \cite{tabak2014}, \cite{Minoiu2013}). In this way, we consider each stock in the portfolio by means of  two  major characteristics: its individual volatility and its interconnection with the system. This strategy acts differently from the classical GMV model that is instead based on the pairwise correlation between assets. \\
The proposed approach is then extended to account also for different dependence structures. Alternatively, a Kendall rank correlation network and a non-linear lower tail dependence network are defined. The optimal problem is then applied by considering two alternative interconnectedness matrices. 
\\

In order to validate our models we perform an empirical analysis where portfolios obtained through network structures will be compared with their analogue GMV. Moreover, the comparison is also extended to the equally weighted portfolio, that has attracted the attention of practitioners, staring with the paper \cite{DeMiguel07}. 
As a robustness check, the proposed approach is applied using empirical data of three different high dimensional portfolios and by testing alternative rolling windows. A first dataset considers the $266$ world's largest banks and insurance companies based on market capitalization. The investment universe of a second portfolio is the same as that of the \textit{S \& P} $100$ index. The analysis is also extended to hedge funds, by considering a third portfolio composed of  $64$ hedge funds downloaded from \lq\lq Hedge Fund Research dataset \rq\rq. \\
Furthermore, by means of in-sample and out-of-sample analysis, we investigate the extent to which the network-based approaches can be used to enhance the portfolio selection process.  \\
The results show that the network-based approaches perform very favourably compared not only to its competing approach (global minimum variance), but also to the \lq\lq equally weighted\rq\rq  strategy. In the majority of the cases, both long-term performances and risk-adjusted measures tend to provide the best behaviour at least for one of the network-based approaches. Furthermore, on average, also lower transaction costs are observed. \\
An empirical evidence of the inverse relation between the average clustering coefficient and the normalized herfindahl index is provided as well. In particular, a higher level of interconnection in the financial market leads to a lower portfolio diversification. \\
The paper is organized as follows. Section \ref{S:port_all} is devoted to the description of the asset allocation problem. Precisely, in Subsection \ref{PortGMV} the classical Global Minimum Variance approach is reminded . The structure of three dependence networks we deal with is explained in Subsection \ref{dep_network}. Furthermore, the formulation of the optimal problem we propose is outlined in Subsection \ref{net_decription}. In Section \ref{dataset_empirical}, the description of the datasets and of the empirical procedure used in this paper is reported. Section \ref{results} performs in-sample and out-of-sample analysis for both network-based approaches and classical conventional portfolio strategies. Although the analysis has been developed for all three datasets, for the sake of brevity  only the results referred to the first dataset are reported here and  all details concerning the other two datasets are provided in the supplementary material. Finally, Section \ref{conclusion} concludes and outlines future research lines.

\section{Asset allocation problems}
\label{S:port_all}
In this section we describe the asset allocation problems. At first, the classical portfolio model is reported, then a new approach is proposed using network theory, in order to find an optimal solution for the portfolio selection process.

\subsection{Global Minimum Variance}
\label{PortGMV}

The classical GMV \color{black} approach aims at determining the fractions $x_{i}$ of a given capital to be invested in each asset $i$
belonging to a predetermined basket of assets in order to minimize the risk of the portfolio, identified through its variance. \\
In particular, we assume that $N$ assets are available and we denote with $R_{i}$  the random variable (r.v.) of daily log returns. Let $\textbf{r}=[r_{i}]_{i=1,..,N}$ be the returns' vector observed in a specific time period$/$window ($w$) and $\bm{\Sigma}$ the variance and covariance matrix between assets estimated in the same period. We denote with $\sigma_i^2$ the $i$-th element on the principal diagonal of $\bm{\Sigma}$ and with $\sigma_p^2$ the variance of the portfolio.\\
The strategy, commonly known as GMV, can be defined as:
\begin{equation}
\label{gmv}
\left\{
\begin{array}{ll}
\smallskip
\min\limits_{\mathbf{x}} \ \ \ \sigma_p^2 = \textbf{x}^{T}\bm{\Sigma}\textbf{x}\\ 
\smallskip
\mathbf{e}^T\mathbf{x}=1\\
0 \leq x_i \leq 1, & \hbox{$i=1, \dots,N$,}
\end{array}
\right.
\end{equation}
 where $\mathbf{e}$ is a vector of ones of length $N$ and $\mathbf{x}=[x_{i}]_{i=1,...,N}$ is the vector of the fractions $x_{i}$ of the initial endowment invested in each asset. The first equation requires that the whole capital should be invested (i.e. budget constraint). Conditions $0 \leq x_i \leq 1$, $i=1, \dots,N$ exclude the possibility of short selling\footnote{As well-known, a closed-form solution of the GMV problem exists if short selling is allowed.}.  \\ 
 It is worth nothing that the dependence structure between assets plays a crucial role in a portfolio allocation model. Problem \eqref{gmv} measures the dependence through the Pearson correlation coefficient between each couple of firms, as proposed in \cite{markowitz52}. In this paper, we aim also at exploiting different kind of dependences. In particular, as described in the next Section, we deal with the interconnections between a specific firm and the the whole investment universe of the portfolio. 
 
In addition to classical portfolio theory, a specific strategy that can be considered consists in holding a portfolio having weight $\frac{1}{N}$ in each component (the so-called equally weighted portfolio (EW)). The EW strategy ignores completely the data and does not require any optimization or estimation procedure.
 

\subsection{Portfolio optimization through network approaches}
\label{networkApp}
In this section, the portfolio selection problem is tackled  within a network framework. The financial market is represented through a network, where assets are nodes and edges accounts for the dependence between returns. 
The developed investment strategy benefits from the knowledge of such market dependency structure. Unlike the classical global minimum variance model, based on the pairwise correlation between assets, we provide an objective function that takes into account the interconnectedness of the system. 
As we will see later, this interconnectedness is modelled by means of a specific structural indicator, the clustering coefficient, that considers how much an asset is connected to other assets in the portfolio. \\

Since our analysis is based on the level of dependence between assets, we propose to explore dependence at different \lq\lq observation scales\rq\rq. At first, we develop our study by using a classical linear correlation network (we refer, among others, to \cite{onnela03,giudici2016}). In this case, we will perform a \textit{Pearson Correlation network approach} (PNA) for portfolio optimization. We also make a step forward and we propose two alternative methods, \textit{Kendall Tau network approach} (KNA) and \textit{Tail dependence network approach} (TNA), where we focus on non-linear dependence between assets. \\    
First, the construction  of the dependency networks will be explained and then the portfolio selection through the network approach will be discussed. 

\subsubsection{Dependency network}
\label{dep_network}
A network $G=(V,E)$ consists by a set $V$ of $N$ nodes and a set $E$ of links (edges) between nodes of $V$. 
The edge connecting a pair of nodes $i$ and $j$ is denoted by $(i,j)$; $i$ and $j$ are called adjacent nodes.
If $(j,i)\in E$ whenever $(i,j)\in E$, the network is undirected.
The degree $k_{i}$ of a node $i$ $(i=1,...,n)$ is the number of edges incident to it.  
$G=(V,E)$ can be conveniently represented through its adjacency matrix $\mathbf{A}$ whose elements are $a_{ij}=1$ whenever $(i,j)\in E$ and $0$ otherwise. 
A network is weighted if a weight $w_{ij} \in \mathbb{R}$ is associated to each edge $(i,j)$. In this case, both adjacency relationships between vertices of $G$ and weights on the edges are described by a non negative, real $N$-square matrix $\textbf{W}$ (the weighted adjacency matrix). \\
Moving to our proposal, correlation network is the first case we deal with. 
Generally speaking, this kind of network belongs to similarity-based class of networks, where a weighted edge represents a similarity (but not necessarily a direct interaction) between the two nodes.
Despite a variety of the similarity measures, a widespread choice consists in quantifying the similarity between two elements (nodes) of the system with the Pearson correlation. Formally, a weighted, complete and undirected network $G$ represents the correlations structure, being all assets correlated, with edges' weight equal to the linear correlation between them.  Thus, the elements of the matrix $\mathbf{W}$ are defined as:
\begin{equation}
\label{rho}
w_{ij}=\begin{cases} \rho\left(R_{i},R_{j}\right) & \mbox{if }  i \neq j \\ 0 & \mbox{otherwise}
\end{cases},
\end{equation}
where $\rho\left(R_{i},R_{j}\right)$ is the linear cross-correlation\footnote{Notice that an ultrametric distance can be associated to the correlation coefficient in order to assure that weights range in a limited interval (see, for instance, \cite{Mantegna99,giudici2016,onnela03}). In our case, this transformation does not affect the results in terms of optimal portfolio.}, which captures the linear relationship between two random variable $R_{i}$ and $R_{j}$.\\
Furthermore, it is worth pointing out that in our model we consider all correlations between assets, obtaining then a complete structure. Some alternative approaches have been provided in literature, in order to reduce the complexity of the network. The simplest filtered network is a threshold network, where all correlation coefficients lower than a given threshold are discarded. A second class of filtering methods is based on either MST (see \cite{onnela03}) or PMFG (see \cite{pozzi13}), applied in order to extract only more relevant correlations. \\

As well-known, linear correlation provides only a partial view of assets' dependence, and several papers provided evidence of non-linear dependence in stock returns, as reported in \cite{hinich1985}. Hence, to extend our analysis, we  study two alternative dependency networks, with different choices of weights on the edges: Kendall Tau network and Tail dependence network. 
The Kendall rank correlation coefficient (commonly referred to as Kendall's tau coefficient) is a well-known measure of concordance 
for bivariate random vectors, while the Tail dependence coefficient provides asymptotic measures of the dependence in the
tails of the bivariate distribution. \\ 
Kendall Tau network is characterized by edges' weights based on Kendall's tau coefficient:
\begin{equation}
w_{ij}=\begin{cases} \tau\left(R_{i},R_{j}\right) & \mbox{if }  i \neq j \\ 0 & \mbox{otherwise}
\end{cases}.
\end{equation}
The Kendall's tau coefficient  is computed as:
\begin{equation}
\tau\left(R_{i},R_{j}\right)= \frac{\sum_{h=1}^{n}\sum_{k\neq h} sgn\left(r_{i,h}-r_{i,k}\right)sgn\left(r_{j,h}-r_{j,k}\right)}{n(n-1)} 
\end{equation}
where $\left(r_{i,1},r_{j,1}\right),\left(r_{i,2},r_{j,2}\right),...,\left(r_{i,n},r_{j,n}\right)$ is the set of $n$ observations of the joint random variables $R_{i}$ and $R_{j}$ in a specific time period. \\
This statistic is typically used to measure the ordinal association between two measured quantities. In other words, it computes the ratio of the difference between the number of concordant pairs and the number of discordant pairs to the total number of pair combinations. \\

With the Tail dependency network, we consider a weighted and undirected graph where edges' weights are equal to the lower tail dependences between each couple of assets\footnote{See Embrechts et al. \cite{embrechts2001} for a discussion about the concept of tail dependence in financial applications related to market or credit risk. A generalization of bivariate tail dependence, as defined above, to the multivariate case can be found in \cite{schmidt2006non}.}, in order to consider only co-movements in the lower tail of log-returns distributions:
\begin{equation}
w_{ij}=\begin{cases} \lambda_{l}\left(R_{i},R_{j}\right) & \mbox{if }  i \neq j \\ 0 & \mbox{otherwise}
\end{cases}.
\end{equation}
The lower tail dependence is defined as:
\begin{equation}
 \lambda_{l}\left(R_{i},R_{j}\right)= \lim_{q \rightarrow 0} P\left(R_{j}\leq F^{-1}_{R_{j}}(q)|R_{i}\leq F^{-1}_{R_{i}}(q)\right)
\end{equation}
where $F^{-1}_{R_{i}}$ is the inverse of the cumulative distribution function of random variable $R_{i}$. Lower tail dependence describes the limiting proportion that one marginal is lower than a certain threshold, given that the other marginals has already fallen below that threshold. The study of tail dependence is an interesting tool in order to quantify the amount of dependence in the tail.




\subsubsection{Problem formulation}
\label{net_decription}
%
We enter now into the methodology's explanation describing the optimization problem. For the sake of brevity, we outline the formulation of the problem following Pearson correlation network approach (PNA), but the procedure that can be followed for the other approaches (KNA and TNA) is essentially the same.

 \noindent
At first, the matrix $\bm{\Sigma}$ can be conveniently normalized by defining a new matrix:
$$\bm{\Omega}=\frac{\bm{\Sigma}}{\sum_{i=1}^{N}\sigma_{i}^{2}}.$$

It is easy to observe that $\bm{\Omega}=\bm{\Delta}^T \boldsymbol{\Pi} \boldsymbol{\Delta}$, where $\bm{\Pi}$ is the correlation matrix and $\bm{\Delta}=diag(s_{i})$ is a diagonal matrix with diagonal entries $s_i=\frac{\sigma_{i}}{\sqrt{\sum_{i=1}^{N}\sigma_{i}^{2}}}$. 
Clearly, the element $s_{i}$ considers the contribute of the standard deviation of the returns $i$ with respect to the total standard deviation, computed in case of independence. Indeed, the dependence effect is contained in the correlation matrix $\bm{\Pi}$.

\noindent Since variance matrix $\bm{\Sigma}$ is positive semidefinite, $\bm{\Omega}$ is also positive semidefinite, as $\bm{\Sigma}$ and $\bm{\Omega}$
share the same eigenvalues\footnote{By basic properties of the determinants, the eigenvalues of $\bm{\Omega}$ can be obtained by those of $\bm{\Sigma}$ by a multiplicative factor:
\begin{center}
$\det{(\bm{\Omega}-\lambda\bm{I})}=\det{\left(\frac{\bm{\Sigma}}{\sum_{i=1}^{N}\sigma_{i}^{2}}-\lambda{\mathbf{I}}\right)}=
\det{\left(\frac{\bm{\Sigma}-\left(\sum_{i=1}^{N}\sigma_{i}^{2}\right)\lambda\mathbf{I}}{\sum_{i=1}^{N}\sigma_{i}^{2}}\right)}=
\left(\frac{1}{\sum_{i=1}^{N}\sigma_{i}^{2}}\right)^n\det{\left(\bm{\Sigma}-\left(\sum_{i=1}^{N}\sigma_{i}^{2}\right)\lambda \mathbf{I}\right)}$.
\end{center}}. 


\noindent
Hence, problem (\ref{gmv}) can be reformulated as follows: 
\begin{equation}
\label{netC}
\left\{
\begin{array}{ll}
\smallskip
\min\limits_{\mathbf{x}} \ \ \ \textbf{x}^{T}\bm{\Omega} \textbf{x}\\ \smallskip
\mathbf{e}^T\mathbf{x}=1\\
0 \leq x_i \leq 1, & \hbox{$i=1, \dots,N$,}
\end{array}
\right.
\end{equation}


Being agents (financial institutions e.g. firms, banks, insurances or hedge funds) part of a complex system, where all assets are correlated, if we represent this correlation by means of the Pearson coefficient, then the complete, weighted and undirected network $G=(V,E)$ identifies the correlations' structure, where each pair of nodes/assets $i$ and $j$ is connected by an edge $(i,j)$, in which edges' weights are defined as in (\ref{rho}). 



To measure the level of interconnectivity of an asset with the whole  system, we use a clustering coefficient. The definition of clustering coefficient, introduced by Watts and Strogatz in \cite{watts1998collective}, refers to the number of existing triangles around a node respect to the number of potential ones, providing a quantity suitable to account for how the node is embedded into the structure, in terms of connections. This is extremely important in revealing system behaviours, such as for instance contagion propagation, and several papers (see, for instance, \cite{tabak2014,hu2012network}) argue that clustering coefficient is a good measure of systemic risk. 

However, being our graph complete, the classic clustering coefficient is not computable and we need to adapt the measure to our framework.
The procedure we follow is in line with \cite{mcassey2015clustering}. Specifically, we modify the adjacency matrix fixing a threshold $s\in[-1,1]$
and defining $\mathbf{A}_{s}$, whose elements are:
\begin{equation}
a_{ij}^s = \begin{cases} 1 & \mbox{if }  w_{ij} \geq s \\ 0 & \mbox{otherwise}
\end{cases}.
\end{equation}

\noindent $\mathbf{A}_s$ is the adjacency matrix describing the existing links in the network having weight $w_{ij}$ at or above the threshold $s$ (in this case, $w_{ij}=\rho(R_i,R_j)\geq s $). The idea is to capture the mean cluster prevalence of the network
looking at a zoom-out level where only the strongest edges (i.e.
greater than a given threshold) are visible. 

\noindent Hence, we compute the clustering coefficient proposed in \cite{watts1998collective}
and then we repeat the process, varying the threshold $s$. 
The clustering coefficient $C_{i}$ for a node $i$ corresponding to
the graph is the average of $C_{i}(\mathbf{A}_{s})$ overall
$s\in[-1,1]$:
\begin{equation}\label{clust_tot}
C_{i}=\int_{-1}^{1}C_{i}(\mathbf{A}_{s})ds
\end{equation}

Since $0\leq C_i\leq1$, $C_{i}$ is well-defined.  
Now, we define the $N$-square matrix $\mathbf{C}$, whose elements are:
\begin{equation}
c_{ij} = \begin{cases} C_{i}C_{j} & \mbox{if }  i \neq j \\ 1 & \mbox{otherwise}
\end{cases}.
\end{equation}

\noindent Since this matrix accounts for the interconnection level of each couple with the system, instead of measuring classical pairwise correlation, we can interpret it as an interconnectedness matrix. 

\begin{proposition}
Let $\mathbf{C}$ be defined as above and $\mathbf{H}=\bm{\Delta}^T\mathbf{C}\bm{\Delta}$. Then 
$\mathbf{C}$ and $\mathbf{H}$ are positive semidefinite.
If at least a couple of assets $k,j$ exists such that $\rho\left(R_{k},R_{j}\right)< 1$, then $\mathbf{C}$ and $\mathbf{H}$ are positive definite. 
\end{proposition}

\begin{proof}
Being $0\leq C_i\leq1$, then, $\forall\mathbf{x}\in R^N$:
\begin{equation}\label{semidef}
\begin{split}
	\mathbf{x}^{T}\mathbf{C}\mathbf{x}=\sum_{i=1}^N x_i^2+\sum_{i,j}^NC_iC_jx_ix_j=
	\sum_{i=1}^N x_i^2+2\sum_{i<j}^NC_iC_jx_ix_j \geq \\
	 \sum_{i=1}^NC_i^2x_i^2+2\sum_{i<j}^NC_iC_jx_ix_j=\left( \sum_{i=1}^NC_ix_i\right)^2\geq 0,
\end{split}
\end{equation}
hence $\mathbf{C}$ is positive semidefinite. 

\noindent Suppose that $\rho\left(R_{k},R_{j}\right)< 1$ for a pair of assets $k,j$. By the definition of matrix $\mathbf{A_s}$, a value of the threshold $\bar{s}$ exists such that $a_{kj}^{\bar{s}}=0$. This implies that the edge $(k,j)$ of the corresponding graph is deleted.\\
Considering that the original graph is complete, from (\ref{clust_tot}), $C_i<1$ $ \forall i$. As a consequence, $x_i^2 > C_i^2x_i^2$ $\forall i$, yielding in (\ref{semidef}) to
$\mathbf{x}^{T}\mathbf{C}\mathbf{x}>\left(\sum_{i=1}^NC_ix_i\right)^2\geq 0$ $\forall \mathbf{x} \neq \mathbf{0}$, 
 concluding that the matrices $\mathbf{C}$, and hence $\mathbf{H}$, are positive definite.
\end{proof}

The optimal problem (\ref{netC}) can be modified as:

\begin{equation}
\label{Ouroptimalconst}
\left\{
\begin{array}{ll}
\smallskip
\min\limits_{\mathbf{x}} \ \ \ \textbf{x}^{T}\mathbf{H}\textbf{x}\\ \smallskip
\mathbf{e}^T\mathbf{x}=1\\
0 \leq x_i \leq 1, & \hbox{$i=1, \dots,N$.}
\end{array}
\right.
\end{equation}

\noindent Observe that the objective function is continuous in a feasible compact set, then at least one solution of the problem exists.
As in classical GMV, we are considering in problem (\ref{Ouroptimalconst}) both volatility of assets and the dependence structures between them. The main difference is due to the use of the interconnectedness matrix in order to consider how much each couple of assets is related to the system. In particular, being $\mathbf{C}$ dependent on a network-based measure of systemic risk (i.e. the clustering coefficient), we are implicitly including a measure of the state of stress of the financial system in the time period. 

\noindent As in \cite{markowitz52}, when short selling is allowed, we are able to provide a closed form solution given by $\mathbf{x}\mathbf{=}\frac{\mathbf{H}^{-1}\mathbf{e}}{\mathbf{e}^{T}\mathbf{H}^{-1}\mathbf{e}}$.

 \color{black}

\section{Dataset description and empirical protocol}\label{dataset_empirical}

In this section, we perform some empirical applications in order to assess the effectiveness of the proposed approaches. Since the aim is to test the robustness and to avoid data mining bias, we analyse three portfolios with different characteristics, arising from three different datasets.   

The first one consists of 266 among largest banks and insurance companies in the world. In particular, we take the greatest firms by market capitalization in the banking and insurance sector. Hence, the dataset contains daily returns of 120 insurers and 144 banks in the time-period ranging from January 2001 to the end of 2017. \\
In addition, we collected daily returns of a second dataset referred to the same time period, that includes $102$ leading U.S. stocks constituents of the $S\&P$ 100 index at the end of 2017. Both data have been downloaded from Bloomberg \cite{bloombe}.\\
\color{black}The third portfolio is composed by $64$ hedge funds downloaded from \lq\lq Hedge Fund Research Dataset\rq\rq and cover the period from\footnote{https://www.hedgefundresearch.com/hfr-database. Observations before than $1^{st}$ of April 2003 are not available for the hedge funds under analysis.} 01-January-2004 to the end of 2017. \\
Although the analysis has been performed for all three datasets, for the sake of brevity we emphasize here the developments referred to the first one, i.e. banks and insurers, and we provide all details concerning the other two datasets in the supplementary material.

In the literature, several criteria exists in order to evaluate portfolio strategies. In this article, we focus on portfolio diversification, on transaction costs and on risk-return performance measures, as will be explained in the following subsections. All these aspects are investigated through a rolling window methodology, which is characterized by an in-sample period of length $n$ and an out-of-sample period of length $h$. \\
Results related to diversification and risk-performance measure are obtained, respectively, by means of in-sample and out-of-sample analysis. The data of the first in-sample window of width $n$ (i.e. from $t=1$ to $t=n$) are used to estimate the optimal weights in the first window. Optimal weights are then invested in the out-of-sample period, from $t=n+1$ to  $t=n+h$, where the out-of-sample performance is computed.  
The process is repeated rolling the window $h$ steps forward. Hence, weights are updated by solving the optimal allocation problem in the new subsample, diversification and portfolio turnover are calculated (i.e. using data from $t=h+1$ to $n+h$ and discarding the first $h$ observations), and the performance is estimated once again using data from $t=n+h+1$ to $n+2h$. Repeating these steps until the end of the dataset is reached, we buy-and-hold the portfolios and we record out-of-sample performance in each rebalancing period.

\noindent
In this paper, we mainly consider monthly stepped two-year windows. 
This choice is a trade off between too noisy and too smoothed data for small and large window widths, respectively. However, in order to test the effect of a different window width $n$, we also perform a rolling window strategy with an in-sample-period of six months and an out-of-sample period of one month length. 

\subsection{In-sample analysis: diversification and transaction costs}
In an in-sample perspective we analyse the obtained portfolios in terms of diversification and transaction costs using, respectively, the modified Herfindahl index and  the portfolio turnover. The modified Herfindahl index on a window $w$ (denoted with $HI_w$) is a normalized version of classical Herfindahl index and it is formally defined as:
\begin{equation*}
	HI_w=\frac{\mathbf{x}_w^{\star^{T}}\mathbf{x}^{\star}_w-\frac{1}{N}}{1-\frac{1}{N}},
\end{equation*}
where $\mathbf{x}^{\star}_w=\left[x^{\star}_i(w)\right]_{i=1,\ldots,N}$ stands for the vector of optimal weights on window $w$. 
Because it is a normalized version of the classical index, it ranges from $0$ to $1$, being equal to $0$, in case of the EW portfolio (the most diversified portfolio) and to $1$ in case of a portfolio concentrated in only one asset. 

A proxy of transaction costs of a given strategy is the portfolio turnover at window $w$ ($\vartheta(w)$) computed as:

\begin{equation*}
	\vartheta(w)=\sum_{i=1}^N{\left|x^{\star}_i(w)-x^{\star}_i(w^{-}))\right|},
\end{equation*}
where $x^{\star}_i(w^{-})$ and $x^{\star}_i(w)$ are the optimal portfolio weights of the $i^{th}$ asset before and after rebalancing (i.e. according to the optimization strategy) at window  $w$, respectively.

\subsection{Out-of-sample analysis: performance measurements}
\label{PerfMeas}

In financial framework, risk-adjusted performance measures are used to compare the absolute returns or the relative returns (i.e. excess returns) to the risk taken.
In this paper, we examine the performance of 
the three datasets using the \textsl{Sharpe Ratio} $(SR)$, the \textsl{Information Ratio} $(IR)$ and the \textsl{Omega Ratio} ($OR$).
\newline
The \textsl{Sharpe Ratio} is defined as: 
\begin{equation*}
	SR=\frac{\mu_p^{\star}- \mu_f}{\sigma_p^{\star}},
	\label{Shar} 
\end{equation*}
where $\mu_p^{\star}$ and $\sigma_p^{\star}$ are  respectively the average return and the standard deviation of the optimal portfolio/strategy and $\mu_f$ indicates the average risk-free rate\footnote{For the sake of simplicity, we set the average risk-free rate at zero in the empirical analysis.}.
This ratio measures the mean  excess return per unit of overall risk.  Typically, the portfolio with the highest Sharpe ratio is regarded as the best portfolio according to this criterion. 
We use the Sharpe ratio to rank different portfolios only in case of positive portfolio excess return with respect to the risk free rate. Indeed, in line with the literature, in this case the Sharpe ratio is considered as an appropriate risk adjusted performance measure\footnote{In case of negative average portfolio excess return this measure is not appropriate and different modifications have been proposed in literature (see \cite{Scholz2007}).}. 

The rationale behind the \textsl{Information Ratio} is similar to that of the \textsl{Sharpe Ratio}, but the average portfolio return is compared with a benchmark portfolio return instead of the risk-free rate. \textsl{IR} is defined as: 
\begin{equation*}
IR=\frac{\mu_{\delta}}{\sigma_{\delta}},
\label{Shar} 
\end{equation*}
where  $\delta=r^{\star}_p-r^{\star}_{GMV}$. \\
$r^{\star}_p$ and $r^{\star}_{GMV}$ represent the out of sample time series of optimal portfolio returns corresponding to a strategy $p$ and to a GMV strategy respectively.
A value of the Information Ratio greater than zero means that strategy $p$ is over performing the benchmark.\\
Since network approaches represent viable alternatives to the classic GMV, a natural choice is to compare our results with GMV strategy. To this reason, we assume GMV as benchmark.  
Therefore it seems quite natural to measure how much the network-based methodology performs better (or worse) with respect to the classic problem (\ref{gmv}).\color{black}

\noindent
The \textsl{Omega Ratio} has been introduced by Keating and Shadwick in \cite{keating2002universal} and it is defined as:
\begin{equation*}
	OR = \frac{\int_{\epsilon}^{+\infty} (1-F(x))\,dx}{\int_{-\infty}^{\epsilon} F(x)dx}=\frac{\mathbb{E}\left(r^{\star}_p-\epsilon\right)^+}{\mathbb{E}\left(\epsilon-r^{\star}_p\right)^+},
\end{equation*}
where $F(x)$ is the cumulative distribution function of the portfolio returns  
and $\epsilon$ is a specified threshold\footnote{We point out that $OR$ ratio is very sensitive to values of $\epsilon$ which can be different from investor to investor. In the empirical analysis $\epsilon$ is set equal to $0$.}. 
\color{black} Returns below the specific threshold are considered as losses and returns above as gains. In general, an OR greater than 1 indicates that  strategy $p$ provides more expected gains than expected losses. The portfolio with the highest ratio will be preferred by an investor. The $OR$ implicitly embodies all the moments of the return distribution without any a-priori assumption. 

\section{Results and discussion}
\label{results}

As previously pointed out, the main contribution of this paper to the existing literature is the proposal of a network-based approach as an effective tool in enhancing the portfolio selection process. As explained in Section \ref{networkApp}, in order to capture the dependence structure between assets, we test three different approaches (PNA, KNA and TNA) and we provide a comparison, taking as benchmark the GMV approach. For the sake of completeness, the EW strategy is also considered.
\raggedbottom
\subsection{Banks and Insurers dataset with monthly stepped two-year rolling windows}

We initially focus on Banks and Insurers dataset where returns have been split by using monthly stepped two-year windows. To give a preliminary idea of the network approach, Figure \ref{F:BI2Ynet} depicts the correlation network obtained in the last window $w$ that covers the period December 2015-November 2017. As previously pointed out, each node represents a firm (bank or insurer) and the weighted edge $\left(i, j\right)$ measures the correlation between firms $i$ and $j$. In Figure \ref{F:BI2Ynet}, bullets size is proportional to the standard deviation of each firm, providing a visual picture of the contribution of each asset to  the matrix $\bm{\Delta}$. Edges opacity is instead proportional to weights. Hence, it shows the pairwise correlation coefficient between each couple of assets. \\
We worth mentioning that the estimation method applied for the correlation matrix is an important issue in building the financial network.  Of course, different estimation methods can be used, such as sample, parametric, shrinkage estimators, to mention a few. However, since the scope of the paper is not to analyse the different estimation techniques, in performing the empirical analysis we  use the sample estimation for the three dependence matrices. Building the network-based portfolio selection model with different estimation method is left for a future research. \\
As described in Section \ref{net_decription}, correlation network has been used in order to assess the community structure of each asset by means of a clustering coefficient. Hence, the matrix $\mathbf{H}$ can be computed and problem (\ref{Ouroptimalconst}) can be solved. We report in Figure
\ref{F:BI2Yopt} the optimal solution of PNA for the same window $w$ that covers the period December 2015-November 2017. In this network representation, size of bullets is instead proportional to allocated weight $x^*_{i}(w)$. We observe that the initial endowment is invested in only 26 firms, 10 banks and 16 insurance companies. However, roughly 94\% of the total amount is invested in insurers, that, in this time period, are characterized on average by both a lower volatility and a lower clustering coefficient. In particular, it is noticeable the case of two insurers (Nationwide Mutual Insurance Company and One America), that are characterized by the lowest standard deviations between the firms and by a high proportion of negative pairwise correlations (for instance, approximatively 90\% of correlations between Nationwide Mutual and other firms is lower than zero). As expected the optimal portfolio is based on a high proportion of the initial endowment invested in these two firms (54\% and 17\% respectively).

\begin{figure}[!h]
	\centering
	\includegraphics[scale=0.4]{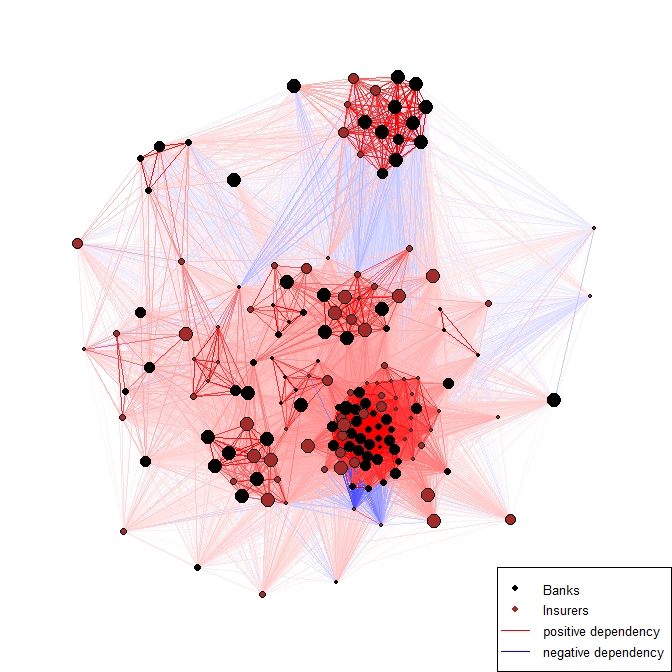}
	\caption{Pearson Correlation Network computed by using returns of Banks and Insurers dataset. Figure refers to the the last window $w$ that goes from the $1^{st}$ of December 2015 to the end of November 2017. Bullets size is proportional to the standard deviation of each firm. Edges opacity is proportional to weights (i.e., intensity of correlation).}
	\label{F:BI2Ynet}
\end{figure}

\begin{figure}[!h]
	\centering
	\includegraphics[scale=0.4]{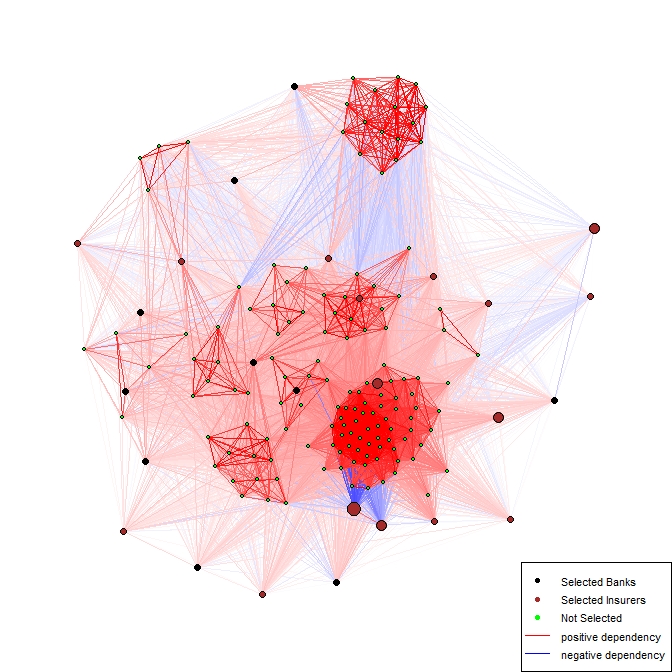}
	\caption{Optimal portfolio for the last in-sample period (December 2015-November 2017) computed by means of PNA. Bullets size is proportional to allocated weight $x_{i}$. Edges opacity is proportional to weights.}
	\label{F:BI2Yopt}
\end{figure}

Now, we perform rolling network-based approaches by considering 180 sub-samples (i.e. 180 rolling windows) of the initial dataset. By means of the optimal weights obtained by each approach, we are able to compute the portfolio diversification and the portfolio turnover. Moreover, by investing the optimal weights in the out-of-sample periods, we also provide different risk and return measures. Main results are reported in Figure \ref{F:BI2Y} and in Table \ref{T:BI2Y1M}.
Concerning the performance, we observe in Figure \ref{F:perBI2_1} that GMV provides the worst behavior in the period under analysis. Although PNA leads to almost the same amount at the end of the period, it always shows a performance higher than GMV, with a significant over-performance in the out-of-sample windows starting in 2006-2007 and since the end of 2009 to the end of 2014. KNA and TNA better catch the dependence structure, leading to an improvement of the performance over the long run. \\
Furthermore, it is interesting to note in Figure \ref{F:TunBI2_1} that network-based approaches are also characterized on average by a lower turnover. It means that allocated weights have been rebalanced less than GMV approach. In general, a strategy with a high turnover rate increases the costs to its investors. The importance of this index is obviously related to the cost for the turnover. \\
We also aim at making a sort of bridge between an in-sample indicator, the diversification index, and the network index we deal with (i.e. the clustering coefficient). We show in Figure \ref{F:ModHerBI2_1} how, on this specific dataset, the optimal network portfolios are more diversified than the GMV  one. Moverover, by comparing these patterns with the average clustering coefficients reported in Figure \ref{F:averclBI2_1} we observe that the higher is the average clustering coefficient the less diversified is the portfolio. This empirical evidence seems reasonable as high average clustering coefficients means high dependence which leads to a lower degree of diversification.

\begin{figure}[!h]%
	\centering
	\subfloat[Out-of-sample performance]{\label{F:perBI2_1}\includegraphics[scale=0.32]{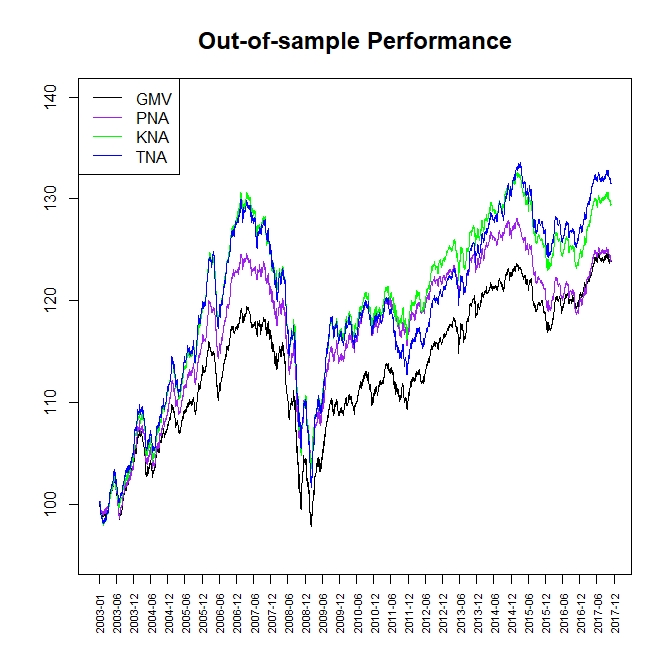}} \hfill
	\subfloat[Portfolio Turnover]{\label{F:TunBI2_1}\includegraphics[scale=0.32]{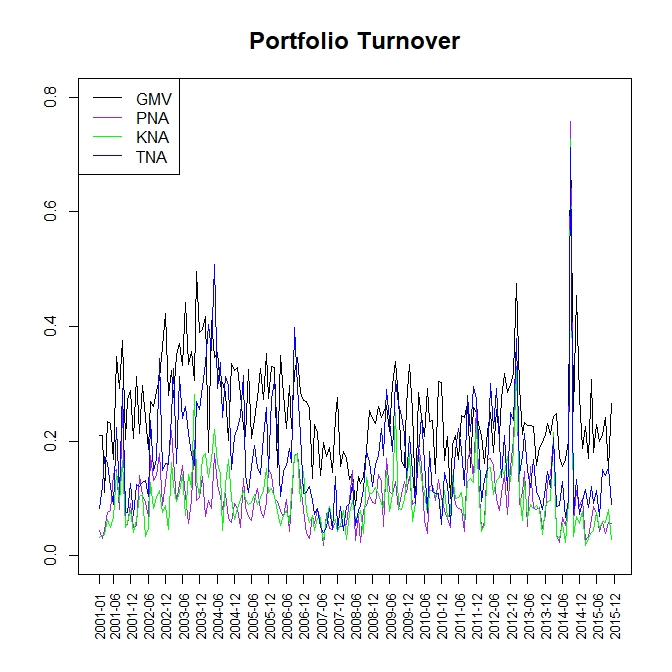}} \\
	\subfloat[Modified Herfindahl Index]{\label{F:ModHerBI2_1}\includegraphics[scale=0.32]{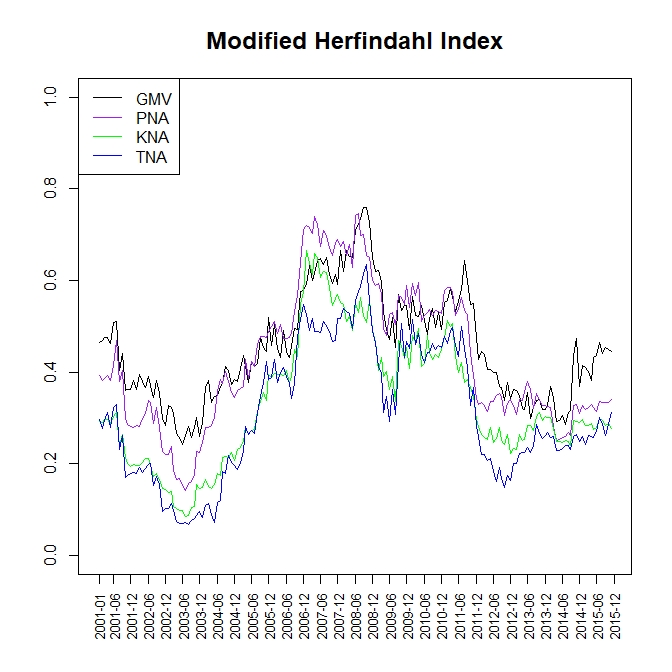}} \hfill
	\subfloat[Average clustering coefficients]{\label{F:averclBI2_1}\includegraphics[scale=0.32]{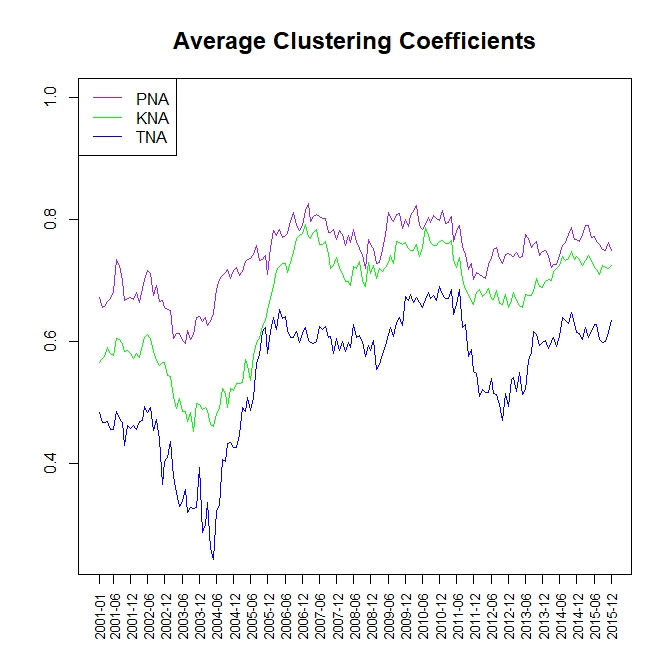}}
	\caption{Banks \& Insurers portfolio results, with 2 years in-sample and 1 month out-of-sample. The dates on the $x$-axes represent the initial dates of the in-sample (out-of-sample) windows $w$ }
	\label{F:BI2Y}
\end{figure}

To complete our analysis, we report in Table \ref{T:BI2Y1M} some basic out of sample statistics and the risk-adjusted performance measures
described in Section \ref{PerfMeas}. According to the sample statistics, it is often difficult to identify the best strategy. For instance, we observe that in general the network-based portfolios have a higher mean but, at the same time, a higher variance with respect to the GMV. This implies that investors with different risk attitude will prefer different portfolios. To overcome this issue we analyse the risk-adjusted performance measures.  The results obtained from Sharpe, Omega and Information Ratio indicate that the two network approaches KNA and TNA overperform the GMV. In particular, the approach, that implicitly includes the lower tail dependences of the bivariate distributions of returns, shows the best performance according to all three risk and returns measures. For the sake of completeness, we provide also the results obtained by using the EW strategy. As observed, this approach leads to a very high volatility without a significant improvement in terms of expected returns.

\begin{table}[!h]
	{\footnotesize \resizebox{0.9\textwidth}{!}{
			\centering
			\begin{tabular}{r|rrrrrrr}
				&    $\mu_p^{\star}$  &    $\sigma_p^{\star}$  &   skewness &   kurtosis &     $SR$ &      $OR$ & $IR$ \\
				\hline
				GMV &      0.015 &      0.029 &     -0.448 &      9.875 &      0.525  &      1.100 &            \\

				PNA &      0.015 &      0.030 &     -0.517 &      8.569 &     0.504  &      1.093 &      0.000 \\

				KNA &      0.018 &      0.031 &     -0.449 &      9.669 &   0.593 &      1.112 &      0.013 \\
				
				TNA&      0.019 &      0.030 &     -0.436 &      6.561 & 0.637 &      1.118 &      0.015 \\
				
				EW &      0.015 &      0.133 &     -0.725 &     22.606 & 0.117 &      1.026 &      0.000 \\
				\hline
			\end{tabular}  
			
	}}
	\caption{Out-of-sample statistics for Banks \& Insurers portfolio with rolling window 2 years in-sample and 1 month out-of-sample. The statistics $\mu_p^{\star}$  and    $\sigma_p^{\star}$ are reported on annual bases. 
	}
	\label{T:BI2Y1M}
\end{table}

\subsection{Alternative datasets and different rolling windows}

The Banks and Insurers portfolio has also been analyzed considering a different rolling window. In particular, we report main results obtained by using monthly stepped six months windows in Figure \ref{F:BI6M} and in Table \ref{T:BI6m1M}. \\
Regarding the performances, we observe that all the considered approaches provide worse results when a shorter rolling window is considered. 
However, also in this case, network-based approaches display a better performance than GMV and EW. Except for the period 2007-2009, we observe indeed a significant over-performance with respect to classical strategies. Furthermore, it is interesting to note that the different kind of dependence structure does not affect the performance in a significant way. This is partially justified by the similar pattern over time of the three alternative dependence measures we deal with. Indeed, it could be noticed in Figure \ref{F:BI6M_4} how the clustering coefficients derived by Pearson and Kendall networks are almost overlapping. Also the clustering coefficient, based on lower tail dependence, is closer to the other ones than in the case of a two-years rolling window. \\
We also observe that reducing the width $n$ of each window, we obtain clustering coefficients more stable over time and, when non-linear dependence is taken into account, higher, on average, than the previous case. As expected, these specific behaviours lead to a lower diversification (see Figure \ref{F:BI6M_3}). Furthermore, it is confirmed that network-based approaches tend to provide on average lower costs related to the portfolio turnover. 

\begin{figure}[!h]%
	\subfloat[Out-of-sample performance]{\label{F:BI6M_1}\includegraphics[scale=0.32]{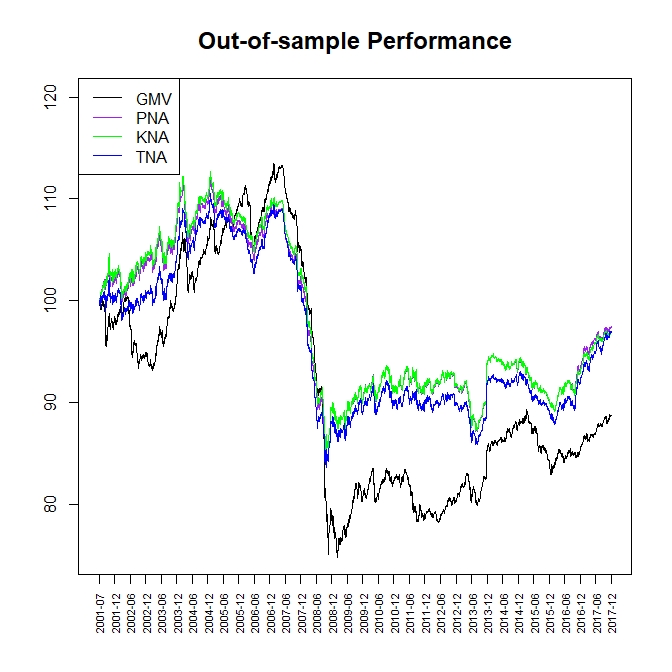}} \hfill
	\subfloat[Portfolio Turnover]{\includegraphics[scale=0.32]{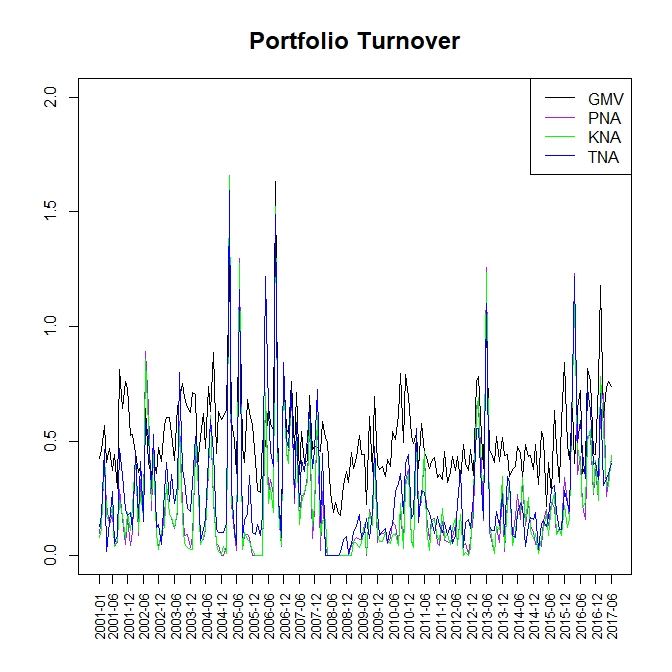}} \\
	\subfloat[Modified Herfindahl Index]{\label{F:BI6M_3}\includegraphics[scale=0.32]{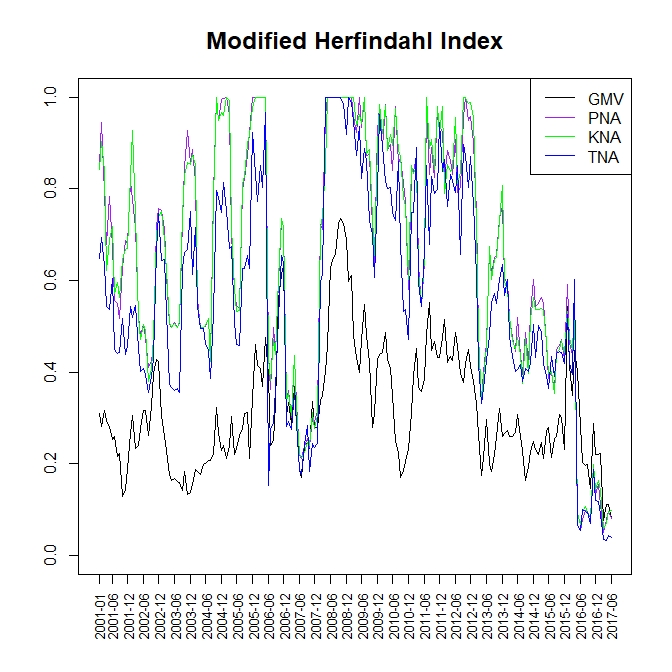}} \hfill
	\subfloat[Average clustering coefficients]{\label{F:BI6M_4}\includegraphics[scale=0.32]{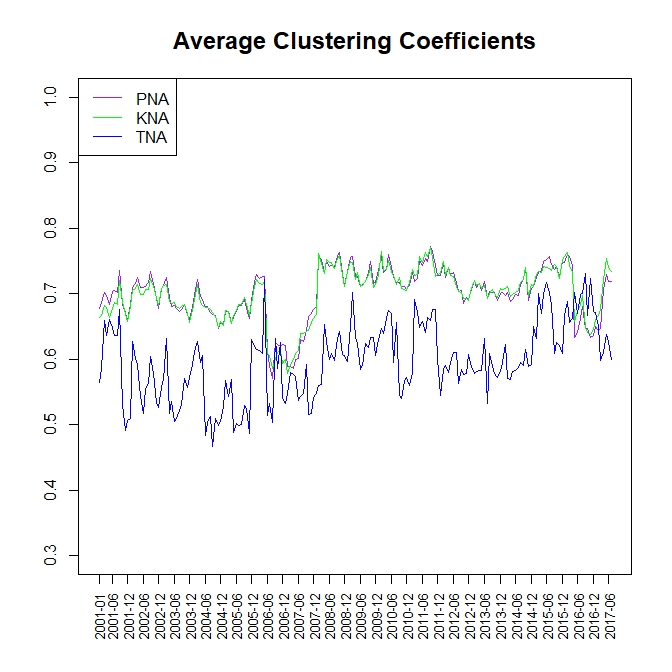}}
	\caption{Banks \& Insurers portfolio results, with 6 months in-sample and 1 month out-of-sample }
	\label{F:BI6M}
\end{figure}

\noindent
Concerning Table \ref{T:BI6m1M}, we have already underlined that, in case of negative excess returns, the Sharpe ratio is not an appropriate measure for ranking the portfolios. 
Hence, in Table \ref{T:BI6m1M} the Sharpe ratio is not reported and the portfolios are compared only in terms of $OR$ and $IR$. Both risk-adjusted performances confirm the over-performance of the three network-based portfolios, with a very close behaviour, with respect to both GMV and EW. Furthermore, as a consequence of Figure \ref{F:BI6M_1}, the Omega Ratio is lower than $1$ because the expected gains are lower than the expected losses.  \\
As already pointed out, the choice of the rolling window has a significant effect on the results. In this specific dataset, a shorter width $n$ seems not able to catch the differences in terms of dependence structure between the time periods. We notice that the lower diversification and the greater turnover do not appear as optimal solutions in this case, leading to worse performances with respect to Table \ref{T:BI2Y1M} for all the strategies. 

\noindent
These findings are in line with the results provided in \cite{cesarone2013learning}, where the authors, through empirical analysis, conclude that the best results for equities portfolios are obtained when the length of the in-sample window ranges between $1$ year and $2.5$ years.

\begin{table}[H]
	{\footnotesize \resizebox{0.8\textwidth}{!}{
			\centering
			\begin{tabular}{r|rrrrrr}
				&    $\mu_p^{\star}$  &    $\sigma_p^{\star}$  &   skewness &   kurtosis &          $OR$ & $IR$ \\
				\hline
				GMV &     -0.007 &      0.032 &     -0.920 &     11.774 &           0.961 &            \\

				PNA &     -0.001 &      0.033 &     -0.104 &      9.939 &           0.995 &      0.014 \\
				
				KNA &     -0.001 &      0.033 &     -0.114 &      9.787 &           0.993 &      0.014 \\
				
				TNA&     -0.001 &      0.032 &     -0.191 &      9.776 &          0.993 &      0.015 \\
				
				EW &     -0.019 &      0.122 &     -0.380 &     13.289 &            0.971 &     -0.007 \\
				\hline
			\end{tabular} 
	}}
	\caption{Out-of-sample statistics for Banks \& Insurers with rolling window 6 months in-sample and 1 month out-of-sample. The statistics $\mu_p^{\star}$  and    $\sigma_p^{\star}$ are reported on annual bases. Values of $SR$ are not reported because of $\mu_p^{\star}$ lower than zero.
	}
	\label{T:BI6m1M}
\end{table}

Now, we summarize main insights observed on the other two datasets (\textit{S\& P} and Hedge Funds) we have explored\footnote{We remind that detailed results are reported in the Supplementary Material}. The results obtained in case of the \textit{S\&P} 100 portfolio are in line with the previous dataset. In particular, independently from the rolling window strategy used, the network-based portfolios have higher out-of-sample performances than the benchmark one (GMV). It is noteworthy that, in this dataset, we obtain the highest average return by applying the EW strategy that does not depend on data. However, despite the best performance, EW is also characterized by a very high volatility and a by the lowest negative skewness. Hence, according to all three risk-adjusted measures, EW strategy appears as the worse choice. \\
Results are less affected by the width of the rolling window.  However also in this case, a longer in-sample period (2 years) leads to best performances and to higher values of the risk-adjusted measures. \\

As regard to the hedge fund portfolio, we obtain slightly different results. The width $n$ of the rolling window strategy is still important but differently to the equities portfolios, better results are obtained for shorter in-sample periods (i.e. six months). The effect of $n$ is in line with the results derived in \cite{hitajBeta}.  \\
However, also for this dataset, the network-based approaches over-perform the GMV strategy, in terms of both average performance and risk-adjusted measures. Finally, it is noticeable that the EW strategy appears as the best solution when a two-year rolling window is considered. 



\section{Conclusions}
\label{conclusion}

In this paper the portfolio selection problem is tackled by means of a network perspective. This work proposes the use of the network topology to catch the interconnectedness of each asset and to provide meaningful insights in portfolio selection process. In particular, we define a new optimal problem that considers both the individual volatility of each asset and how much each node is embedded in the system. The latter issue is caught by means of a local clustering coefficient. Being this coefficient widely recognized in the literature as a network-based measure of systemic risk, we are implicitly including a measure of the financial distress of the system. 
Hence, the proposed approach extends classical GMV strategy by considering a dependence structure based on correlation network, which is different to the classical pairwise linear correlation. \\
Additionally, we develop similar approaches by also considering alternative dependence measures. In this way, we provide a view of assets’ dependence at different observation scales by also including the effect of non-linear dependence in stock returns. \\
To test the robustness of our approaches, a numerical analysis, based on different dataset and alternative rolling windows, has been performed. 
The results obtained are promising, and show that the use of network-based portfolios, in the majority of the cases, leads to higher out-of-sample risk-adjusted performances compared to the global minimum variance and to the equally weighted portfolios. \\
The results obtained encourage the develop of further research, on  incorporating and extending the use of network theory on portfolio selection models. We believe that the network approach considered in this paper can be extended to the so called Smart Beta strategies such as \textsl{equally risk contribution} and \textsl{maximum diversified portfolio}, proposed in \cite{Maillard60} and \cite{Choueifaty} respectively,  where, instead of the covariance matrix, the dependence structure matrix of the network can be used.  
Moreover, it is evident that the estimation method of the dependence structure used to build  the network plays an important role. 
To this end, it is interesting to understand and analyse the impact of different estimation methods on the network-based obtained portfolios. 
We argue that the proposed network-based portfolio selection model could be a valuable contribution to the growing literature on high-dimensional portfolio allocation analysis.


\bibliographystyle{spmpsci}      
\bibliography{Myref}   

\end{document}